\pdfoutput=1
\RequirePackage[fleqn]{amsmath}
\documentclass[runningheads,anycase]{llncs}

\usepackage{style}
\usepackage{shortcuts}

\def\epsilon{\varepsilon}
\def\eps{\varepsilon}

\begin{document}

\title{Faster Combinatorial \texorpdfstring{$k$}{k}-Clique Algorithms\thanks{This work is part of the project CONJEXITY that has received funding from the European Research Council (ERC) under the European Union's Horizon Europe research and innovation programme (grant agreement No.~101078482). The first author is additionally supported by an Alon scholarship and a research grant from the Center for New Scientists at the Weizmann Institute of Science.}}
\titlerunning{Faster Combinatorial $k$-Clique}

\author{Amir Abboud\and Nick Fischer\and Yarin Shechter}
\authorrunning{A. Abboud, N. Fischer, Y. Shechter}
\institute{Weizmann Institute of Science\\Rehovot, Israel\\\email{\{amir.abboud,nick.fischer,yarin.shechter\}@weizmann.ac.il}}


\maketitle

\begin{abstract}
Detecting if a graph contains a $k$-Clique is one of the most fundamental problems in computer science.
The asymptotically fastest algorithm runs in time $O(n^{\omega k/3})$, where $\omega$ is the exponent of Boolean matrix multiplication.
To date, this is the only technique capable of beating the trivial $O(n^k)$ bound by a polynomial factor.
Due to this technique's various limitations, much effort has gone into designing ``combinatorial'' algorithms that improve over exhaustive search via other techniques.

The first contribution of this work is a faster combinatorial algorithm for $k$-Clique, improving Vassilevska's bound of $O(n^{k}/\log^{k-1}{n})$ by two log factors.
Technically, our main result is a new reduction from $k$-Clique to Triangle detection that exploits the same divide-and-conquer at the core of recent combinatorial algorithms by Chan (SODA'15) and Yu (ICALP'15).

Our second contribution is exploiting combinatorial techniques to improve the state-of-the-art (even of non-combinatorial algorithms) for generalizations of the $k$-Clique problem.
Specifically we introduce the first non-trivial clique listing algorithms: A general algorithm for detecting cliques in hypergraphs that also improves the state of the art for detection, and an additional algorithm for the specific case of triangle listing with improved bounds.

\end{abstract}

\section{Introduction} 
\label{sec:intro}

One of the most fundamental problems in computer science is $k$-Clique: given an $n$-node graph, decide if there are $k$ nodes that form a clique, i.e. that have all the $\binom{k}{2}$ edges between them.
Our interest is in the case where $3 \leq k \ll n$ is a small constant.
This is the ``SAT of parameterized complexity'' being the canonical problem of the W[1] class of ``fixed parameter intractable'' problems, and its basic nature makes it a core task in countless applications where we seek a small sub-structure defined by pairwise relations.

The na\"{i}ve algorithm checks all subsets of $k$ nodes and runs in $O(k^2 \binom{n}{k})$ time, which is $\Theta(n^k)$ for constant $k$.
Whether and how this bound can be beaten (in terms of worst-case asymptotic time complexity) is a quintessential form of the question: \emph{can we beat exhaustive search?}

The asymptotically fastest algorithms gain a speedup by exploiting fast matrix multiplication -- one of the most powerful techniques for beating exhaustive search.
In particular, for the important special case of $k=3$, i.e. the \emph{Triangle Detection} problem, the running time is $O(n^{\omega})$ where $2 \leq \omega < 2.3719$~\cite{DuanWZ23} is the exponent in the time complexity of multiplying two $n \times n$ binary matrices.\footnote{Simply compute $A^2$ where $A$ is the adjacency matrix of the graph and check if $A^2[i,j]>0$ for any $\{i,j\}$ that are an edge.}
For larger $k>3$, there is a reduction to the $k=3$ case by Ne{\v{s}}et{\v{r}}il and Poljak \cite{NP85} that produces graphs of size $O(n^{\lceil k/3 \rceil})$.\footnote{Each $k/3$-clique becomes a node and edges are defined in a natural way so that a triangle corresponds to a $k$-clique.}
The resulting time bound is $O(n^{\lceil \omega k /3 \rceil})$. Except for improvements for $k$ that is not a multiple of $3$ \cite{EG04}, and the developments in fast matrix multiplication algorithms reducing the value of $\omega$ over the years, this classical algorithm remains the state-of-the-art.

The one general technique underlying all fast matrix multiplication, starting with Strassen's algorithm \cite{Strassen1969GaussianEI}, is to find some clever formula to exploit cancellations in order to replace multiplications with additions.
To date, this is \emph{the only} technique capable of beating exhaustive search by a polynomial $n^{\eps}$ factor for the $k$-Clique problem.
All techniques have their limitations, and so does Strassen's. Consequently, much research has gone into finding ``combinatorial algorithms'' that beat exhaustive search by other techniques (see \cref{sec:combinatorial-discussion} below).
Existing techniques have only led to polylogarithmic speedups, leading the community to the following conjectures that have become the basis for many conditional lower bounds.

\begin{conjecture}[Combinatorial BMM] \label{conj1}
Combinatorial algorithms cannot solve Triangle Detection in time $O(n^{3-\eps})$ where $\eps>0$.\footnote{Note the informality in these combinatorial conjectures stemming from the lack of precise definition for ``combinatorial'' in this context. See full paper for further discussion.}\end{conjecture}

A reduction of Vassilevska and Williams \cite{WilliamsW18} shows that this conjecture is \emph{equivalent} to the classical conjecture that combinatorial algorithms cannot solve Boolean Matrix Multiplication (BMM) in truly subcubic time \cite{lee2002fast,RodittyZ04}. Following their reduction, many conditional lower bounds were based on this conjecture, e.g. \cite{AbboudW14,clifford2018upper,chan2020range,casel2021fine} (we refer to the survey \cite{VassilevskaWilliams18} for a longer list).

\begin{conjecture}[Combinatorial \boldmath$k$-Clique] \label{conj2}
 Combinatorial algorithms cannot solve $k$-Clique in time $O(n^{k-\eps})$ for any $k \geq 3$ and $\eps>0$.
\end{conjecture}
The latter conjecture is stronger than the former, in the sense that faster algorithms for $k=3$ imply faster algorithms for larger $k>3$ but the converse is not known. The first use of this conjecture as a basis for conditional lower bounds was by Chan \cite{Chan08} to prove an $n^{k-o(1)}$ lower bound for a problem in computational geometry. Later, Abboud, Backurs, and Vassilevska Williams \cite{ABV15} used it to prove $n^{3-o(1)}$ lower bounds in P. Several other papers have used it since then, e.g. \cite{Chang16,bringmann2017dichotomy,lincoln2018tight,bringmann2018clique,abboud2017fine,bringmann2020tree,dalirrooyfard2019graph,AbboudGIKPTUW19,bringmann2019fine,bergamaschi2021new,JX22}.

\paragraph{Previous Combinatorial Bounds}
The previous bounds for Triangle detection ($k=3$) fall under three conceptual techniques. We open with an overview of these techniques (see \cref{sec:combinatorial-overview} for a more detailed review). Note that $(\log\log n)$ factors are omitted in this paragraph.

\begin{enumerate}[topsep=\medskipamount, itemsep=\medskipamount]
\item The \emph{Four-Russians technique} \cite{FourRussians70} from 1970 gives an $O(n^3 / \log^2 n)$ bound, and is used in all later developments.
\item In 2010, Bansal and Williams \cite{BansalW12} use \emph{pseudoregular partitions} to shave off an additional $\log^{1/4}n$ factor.
\item In 2014, Chan \cite{Chan15} introduced a simple \emph{divide-and-conquer technique} to get an $O(n^3 / \log^3 n)$ bound, and a year later, Yu \cite{Yu18} optimized this technique to achieve a bound of $O(n^3 / \log^4 n)$.
\end{enumerate}

For $k>3$ there are two options: (1) we either apply these algorithms inside the aforementioned reduction to Triangle, getting a bound of $O(n^k/\log^4{n})$, or (2) we apply these combinatorial techniques directly to $k$-Clique.
An early work of Vassilevska~\cite{Vassilevska09Clique} from 2009 applied the Four-Russians technique directly to get an $O(n^{k}/\log^{k-1})$ bound.
Note that this generalizes the $\log^2$ shaving from the first bullet naturally to all $k$, and is favorable to the algorithms from option (1) for $k>5$.
Vassilevska's bound remains state-of-the-art, and in this work, we address the challenge of generalizing the other combinatorial techniques to $k$-Clique.

\subsection{Our Results}
The first result of this paper is a faster combinatorial algorithm for $k$-Clique for all $k>3$ based on a generalization of the divide-and-conquer technique from Chan's and Yu's algorithms for $k=3$.
We use divide-and-conquer to design a \emph{more efficient reduction} from $k$-Clique to the $k=3$ case.
The main feature of this reduction is that we get an additional log factor shaving each time $k$ increases by one; this should be contrasted with the classical reduction from option(1) above, in which we gain nothing when $k$ grows.

\begin{restatable}[Reduction from \boldmath$k$-Clique to Triangle]{theorem}{thmkcliquetotriangle} \label{thm:k-clique-to-triangle}
Let $k \geq 3$, and let $a, b$ be reals such that there is a combinatorial triangle detection algorithm running in time \smash{$\Order(n^3 (\log n)^a (\log\log n)^b)$}. Then there is a combinatorial $k$-clique detection algorithm in time $\Order(n^k (\log n)^{a - (k-3)} (\log \log n)^{b + k - 3})$.
\end{restatable}

Combining our reduction with Yu's state-of-the-art combinatorial algorithm for Triangle detection, we improve Vassilevska's bound by two log factors.

\begin{restatable}[Faster Combinatorial \boldmath$k$-Clique Detection]{corollary}{corbestkclique} \label{cor:best-k-clique}
There is a $k$-clique detection algorithm running in time $\Order(n^{k} (\log n)^{-(k+1)} (\log\log n)^{k+3})$.
\end{restatable}

It may be interesting to note that our reduction can even be combined with the na\"{i}ve $O(n^3)$ algorithm for Triangle detection, giving a $(\log n)^{k-3}$ shaving for $k$-Clique \emph{without using the Four-Russians technique}. 

Another interesting implication of our reduction is concerning the framework of Bansal and Williams' \cite{BansalW12}. Their algorithm can be improved if better dependencies for regularity/triangle removal lemmas are achieved. The best known upper bound on $f(\varepsilon)$ in a triangle removal lemma is of the form \smash{$\frac{c}{(\log^*(1/\varepsilon))^\delta}$} for some constants $c>1$ and $\delta>0$.\footnote{Fox achieved some improved dependencies with a new proof of the removal lemma~\cite{DBLP:FoxRemove}, however, it is not clear whether it can be implemented efficiently.}. Due to this dependency, their first algorithm~\cite[Theorem~2.1]{BansalW12} only shaves a $\log^*(n)$ factor from the running times achieved with the standard Four-Russians technique. However, it is not ruled out that much better dependencies can be achieved that would accelerate their algorithm to the point where, combined with our reduction, a $k$-clique algorithm with faster running times than \cref{cor:best-k-clique} is obtained.\footnote{Note that the same cannot be said about their second algorithm~\cite[Theorem~2.1]{BansalW12}; see the lower bound for pseudoregular partitions due to Lovasz and Szegedy \cite{Lovsz2007SzemerdisLF}).}

As discussed in \cref{sec:combinatorial-discussion}, a primary reason to seek combinatorial algorithms for $k$-Clique is that the techniques may generalize in ways fast matrix multiplication cannot (see full paper for detailed discussion).
Our second set of results exhibits this phenomenon by shaving logarithmic factors over state-of-the-art for general (non-combinatorial) algorithms.

One limitation of the $O(n^{\omega})$ algorithm for Triangle detection is that it does not solve the \emph{Triangle listing} problem: we cannot specify a parameter $t$ and get all triangles in the graph in time $O(n^{\omega} + t)$ assuming their number is up to~$t$.
Listing triangles in an input graph is not only a natural problem, but it is also connected to the fundamental 3SUM problem (given $n$ numbers, decide if there are three that sum to zero).
A reduction from 3SUM \cite{Patrascu10,KopelowitzPP16} shows that in order to beat the longstanding $O(n^{2}/\log^2{n})$ bound over integers \cite{BaranDP08} it is enough to shave a~$\log^{6+\eps}{n}$ factor for Triangle listing -- i.e., achieve a running time of~\smash{$O(\frac{n^{3}}{\log^{6+\eps}{n}} + t)$} for some $\eps > 0$.
Although research has seen some results on triangle listing \cite{DBLP:conf/icalp/BjorklundPWZ14}, we are not aware of any previous $o(n^3)+O(t)$ time bound for this problem (even with non-combinatorial techniques). Our second result produces such a time bound, showing that the other combinatorial techniques (namely Four-Russians and regularity lemmas) can be exploited. We shave a $\log^{2.25}{n}$ factor for this problem, generalizing the Bansal-Williams bound for BMM.
Note we use the non-standard notation $\Ologlog(n) = n (\log \log n)^{\Order(1)}$ to suppress polyloglog factors.

\begin{restatable}[Faster Triangle Listing]{theorem}{thmtrianglelisting} \label{thm:triangle-listing}
There is a randomized combinatorial algorithm that lists up to $t$ triangles in a given graph in time \smash{$\Ologlog(\frac{n^3}{(\log n)^{2.25}} + t)$}, and succeeds with probability $1 - n^{-100}$. 
\end{restatable}
Another well-known limitation of Strassen-like techniques is that they are ineffective for detecting hypergraph cliques.
They fail to give any speedup even for the first generalization in this direction: detecting a $4$-clique in a $3$-uniform hypergraph (i.e. a hypergraph where each hyper-edge is a set of three nodes).
The conjecture that $O(n^{4-\eps})$ time cannot be achieved has been used to prove conditional lower bounds, e.g. \cite{lincoln2018tight,CK19}.
Our third result is a $\log^{1.5}{n}$ factor shaving for this problem with an algorithm that also extends to listing. This improves the states of the art (even non-combinatorial) for detection that is due to Nagle \cite{DBLP:journals/siamdm/Nagle10} and provides the first non-trivial algorithm for listing. The following theorem provides our general bound for listing (detection can be obtained by setting $t=1$).

\begin{restatable}[Faster \boldmath$k$-Hyperclique Listing]{theorem}{thmkhyperclique} \label{thm:k-hyperclique}
There is an algorithm for listing up to $t$ $k$-hypercliques in an $r$-uniform hypergraph in time
\begin{equation*}
    \Order\parens*{\frac{n^{k}}{(\log n)^{\frac{k-1}{r-1}}} + t}
\end{equation*}
(assuming a word RAM model with word size $w = \Omega(\log n)$).
\end{restatable}

\medskip
\paragraph{Subsequent Work} Shortly after this work, Abboud, Fischer, Kelly, Lovett, and Meka announced a combinatorial algorithm for BMM with $O(\frac{n^{3}}{2^{(\log n)^\eps}})$ running time \cite{AFKLM23}. This implies an improvement for $k$-Clique as well that is stronger than any poly-log speedup and thus improves over Corollary~\ref{cor:best-k-clique} (by using pseudo-regularity techniques rather than divide-and-conquer).
Moreover, building on our proof of Theorem~\ref{thm:triangle-listing} the authors present a speedup for triangle listing as well.
However, our result for hypergraphs in Theorem~\ref{thm:k-hyperclique} remains unbeaten.

\subsection{On Combinatorial Algorithms} \label{sec:combinatorial-discussion}
While a single satisfying definition of ``combinatorial algorithms'' remains elusive, we believe the search for such algorithms to be of great importance and that the research entertaining this loosely defined notion has already been invaluable. To explain this, let us review the limitations of the Strassen-like technique that motivate us to seek other techniques.
The first two motivations in this list are most commonly mentioned in works on the subject; however, we believe that the latter two are no less important.

One could ask for a specific definition of a ``combinatorial'' algorithm satisfying each motivation. Indeed, each such definition would be of some interest, and in fact, it has already been accomplished to some extent (as we discuss below). 
However, satisfying one definition does not guarantee satisfying the others; an algorithm that generalizes in one setting may not generalize in another or may not be practical or elegant. 
At this point of history, in which we are very much interested in \emph{any} new technique breaking Conjectures~\ref{conj1} and~\ref{conj2} and satisfying \emph{any} of the motivations, it is natural that the community prefers to work with the most inclusive definition (``anything but Strassen-like techniques'' or ``cancellation-free'') that is inevitably loose.

\begin{enumerate}[topsep=\medskipamount, itemsep=\medskipamount]
\item \emph{Practical efficiency:} Strassen's algorithm, and especially its successors, are not as efficient in practice as their asymptotic complexity suggests. This has two reasons: (1) the hidden factors are large, and (2) it is not cache-friendly since it frequently needs to access remote parts of the matrices.
It is hoped that achieving comparable running times (and not just log shavings) with combinatorial techniques will lead to gains in the real world.
This concern has motivated some of the classical works on combinatorial algorithms such as  \cite{lee2002fast,AingworthCM96}.
A definition for this specific notion may have to be empirical, i.e. ``does it work well in practice?''

\item \emph{Elegance:} The matrix multiplication algorithms may be considered inelegant and uninterpretable since it is very hard to grasp the relationship between the intermediate values computed throughout the execution with the very simple problem at hand (e.g. detecting a triangle).
This concern is probably at the origin of the name ``combinatorial'' where one imagines algorithms that only deal with natural combinatorial structures of the input graph such as the neighborhoods of nodes.
Formal models of computations that aim to capture this definition have been suggested and analyzed \cite{angluin1976four,DasKS18} (strong lower bounds were obtained).
They are unfortunately not satisfying \emph{yet} because they do not capture the recent algorithms based on divide-and-conquer or regularity partitions that are widely acknowledged as ``combinatorial'' in the same sense they aim to address.

\item \emph{Generalizability:} Fine-grained complexity has isolated the Triangle detection problem as the most basic form of hundreds of problems; several of its generalizations are the subject of conjectures that form the basis of conditional lower bounds for dozens of problems.
Despite decades of efforts, Strassen-like techniques have failed to achieve speedups for some of these generalizations, and one may hope that a different technique breaking Conjecture~\ref {conj1} will also break the corresponding conjectures (that have nothing to do with ``combinatorial algorithms''). 
For example:
\begin{itemize}[topsep=\smallskipamount, itemsep=\smallskipamount]
\item An algorithm that generalizes to \emph{weighted} graphs (where we ask for the minimum weight triangle) would break the \emph{All-Pairs Shortest-Paths Conjecture} (APSP) \cite{WilliamsW18}. A corresponding conjecture for $k>3$ is also popular (e.g. \cite{abboud2014consequences,BackursDT16,BackursT17,bringmann2020tree}).
\item An algorithm that generalizes to hypergraphs (e.g. to find a $4$-clique in a hypergraph) would break the \emph{Hyper-Clique Conjecture} \cite{lincoln2018tight,abboud2018more}.
\item An algorithm that lets us output witnesses efficiently (e.g. to listing all triangles) would break both the \emph{3SUM Conjecture} \cite{Patrascu10,KopelowitzP18} and the APSP Conjecture \cite{VassilevskaWilliamsX20}. Similar connections exist for $k>3$ \cite{virginia_listing_cliques}.
\item An algorithm that generalizes to the online setting (where the nodes are revealed one at a time) would break the \emph{Online Matrix Vector Conjecture} \cite{HenzingerKNS15,LarsenW17}. Jin and Xu have recently introduced a corresponding conjecture for $k>3$ \cite{JX22}.
\end{itemize}
Note that a precise definition of ``combinatorial'' for each specific motivation can readily be made by asking for an algorithm that refutes the problem in parenthesis (and the other conjectures have essentially done this).

\item \emph{Being the right technique:} It is plausible that Triangle detection can be solved in $O(n^{2+o(1)})$ time (in particular, if $\omega=2$) and that $k$-Clique can be solved in $n^{2k/3}$ or $n^{k/2}$ or perhaps even faster.
Achieving that with the Strassen-like technique has been tedious. 
One may hope that if we get a small speedup such as $n^{2.9}$ with the \emph{right} technique, we would quickly be able to optimize it and reach the final goal of $O(n^2)$.
Anecdotally, one could say that such a discovery was recently made for the All-Pairs Max-Flow problem where a 60-year-old bound was broken by an $O(n^{2.8334})$ combinatorial algorithm \cite{AbboudKT21} (for unweighted graphs) and a complete resolution with an $n^{2+o(1)}$ time bound quickly followed \cite{AbboudK0PST22} (even for weighted graphs).

\end{enumerate} 

We emphasize that a ``combinatorial algorithm'' is expected to satisfy at least one of these motivations, but not all of them. Indeed, some of the contributions of this work is to show that some of the combinatorial techniques do generalize to some settings. Specifically, Theorem~\ref{thm:k-hyperclique} shows that the Four-Russians technique applies in \emph{hyper-graphs}, and Theorem~\ref{thm:triangle-listing} shows that regularity partitions can be used to give a \emph{listing} algorithm.

\paragraph{On Combinatorial Lower Bounds}
Finally, we would like to remark on the value of the related work on ``\emph{combinatorial conditional lower bounds}'' that are based on Conjectures~\ref{conj1} and~\ref{conj2}, proved by designing ``combinatorial reductions''.

Imagine you are an algorithm researcher working on a problem $A$ (e.g. RNA Folding).
You wake up in the morning and try technique $T_1$ on $A$.
It does not work, so you try $T_2$, and then $T_3$ and so on up to $T_{100}$ where $100$ represents the number of relevant techniques.
Each time, you may spend days or weeks trying the technique on $A$, searching for a way to fit into the right framework. Perhaps you hire new students to work on this fruitlessly each time.
Now, suppose someone comes and tells you, because of a reduction from problem $B$ ($k$-Clique) to problem $A$, that it is useless to try any of $T_1,...,T_{100}$ on $A$, because all the experts have already tried them on $B$, without any success, and your problem is at least as hard.
The next morning, you wake up knowing that your choices are either (1) solve $A$ with fast matrix multiplication (the only technique found to be effective against $B$), or (2) invent new techniques but try them on $k$-Clique first (since its bare-bones nature compared to $A$ makes it a better test-bed).

Success stories of such lower bounds leading to the discovery of breakthroughs exploiting fast matrix multiplication include the first $n^{3-\eps}$ algorithm for RNA Folding \cite{bringmann2019truly}, the first $n^{k-\eps}$ algorithm for minimum $k$-cut \cite{li2019faster}, and new techniques in pattern matching \cite{bernardini2019even}.

This work contributes to this line of work by providing an algorithmic attack on Conjecture~\ref{conj2}.

\subsection{Outline}
We start with some preliminaries in \cref{sec:preliminaries}. In \cref{sec:k-clique} we provide our improved combinatorial $k$-Clique algorithm. In \cref{sec:triangle-listing,sec:k-hyper-clique} we provide the high-level ideas of our improvements for Triangle Listing and $k$-Hyperclique Detection; due to space constraints we are forced to defer the technical details to the the full paper.

\section{Preliminaries} \label{sec:preliminaries}
Let $[n] := \set{1, \dots , n}$. We write $\Olog(n) = n (\log n)^{\Order(1)}$ to suppress polylogarithmic factors and use the non-standard notation $\Ologlog(n) = n (\log \log n)^{\Order(1)}$.

Throughout we consider undirected, unweighted graphs. In the $k$-clique problem, we are given a $k$-partite graph $(V_1, \dots, V_k, E)$ and the goal is to determine whether there exist $k$ vertices $v_1 \in V_1, \dots, v_k \in V_k$ such that there is an edge $(v_i, v_j)\in E$ for every pair $i \neq j$. Note that the assumption that the input graphs are $k$-partite is without loss of generality, and can be achieved by a trivial transformation of any non-$k$-partite graph $G= (V, E)$: We create $k$ copies $V_1, \dots,V_k$ of the vertex set and for every $(u,v) \in E$ we add the edges $(u_i, v_j)$ for every~\makebox{$i\neq j$}. Another typical relaxation is that we only design an algorithm that \emph{detect} the presence of $k$-cliques (without actually returning one). It is easy to transform a detection algorithm into a finding algorithm using binary search without asymptotic overhead.\footnote{More specifically, any detection algorithm can be transformed into a finding algorithm with constant running time overhead by using binary search as follows: Arbitrarily split each of the k vertex parts into two halves. Then for each subgraph induced by one of the $2^k$ combination of halves whether it contains a $k$-clique. If the detection algorithm succeeds on some combination, we continue on this combination recursively. For any natural running time the recursive overhead becomes a geometric sum and thus is constant.}

We additionally define the following notation for a $k$-partite graph as before: For a vertex $v$, let $N_{i}(v) = \set{ u \in V_{i} : (v,u) \in E}$ denote the neighbourhood of $v$ in $V_{i}$ and $d_{i}(v) = |N_{i}(v)|$ denote the degree of $v$ in $V_{i}$. Moreover, for a vertex set $V'\subseteq V$ we let $G[V']$ denote the subgraph of $G$ induced by the vertex set $V'$. Throughout we further let $n = |V_1| + \dots + |V_k|$ denote the total number of vertices in the graph.

An \emph{$r$-uniform hypergraph} is a pair $(V, E)$, where $V$ is a vertex set and~\makebox{$E \subseteq \binom Vr$} is a set of \emph{hyperedges.} In the \emph{$r$-uniform $k$-hyperclique} problem we need to decide whether in a $k$-partite hypergraph $(V_1, \dots, V_k, E)$ there are vertices $v_1 \in V_1, \dots, v_k \in V_k$ such that all hyperedges on $\set{v_1, \dots, v_k}$ are present. Similarly, the assumption that the hypergraph is a $k$-partite is without loss of generality.

We are using the standard word RAM model with word size $w\in \Theta \log(n)$. In this model a random-access machine can perform arithmetic and bitwise operations on $w$-bit words in constant time.

\subsection{Overview of Previous Combinatorial Algorithms}\label{sec:combinatorial-overview}

\subsubsection{The Four-Russians Method}
We start with presenting the most basic form of the Four-Russians method~\cite{FourRussians70}, and then move on to a specific variant which we will use throughout this paper. In the Four-Russians method, we start by precomputing the solutions for triangle detection on tiny instances. In particular, we precompute and store the solution to triangle detection on every tri-partite graph with vertex sets of size~\makebox{$s = 0.01 \sqrt{\log(n)}$}, say, and store the answers in a lookup table; note that there are only $n^{0.02}$ such graphs. Then, we use this precomputation to accelerate triangle detection on the input graph: We partition every vertex set in the input graph into blocks of size $s$. Then, for each triplet of blocks we query the lookup table with the subgraph induced by the triplet to obtain an answer to whether there exists a triangle among these blocks. Assuming a word-RAM model with~\makebox{$w \geq \Omega(\log n)$} (as is standard), and some additional preprocessing which allows to determine subgraphs on combinations of blocks efficiently, the running time is dominated by the number of queries to the lookup table which is $O(\frac{n^3}{(\log n)^{1.5}})$.

This basic idea can be slightly improved: Given a tri-partite graph $G=((V_{1},V_{2},V_{3}),E)$, we begin by preprocessing the graph: We divide $V_{2}$ and $V_{3}$ into blocks of size $\varepsilon\cdot\log n$ for some small $\varepsilon>0$. Now, for every pair~$(S, S')$ where $S$ is a subset of a block in $V_{2}$ and $S'$ is a subset of a block in $V_{3}$, we check whether there exists an edge between the two sets and store the answer in a lookup table. Again this preprocessing is affordable as each block has $O(n^{\varepsilon})$ subsets and the number of blocks in each vertex set is at most \smash{$O(\frac{n}{\log n})$}. The total number of pairs $(S,S')$ is thus \smash{$O(\frac{n^{2+2\varepsilon}}{(\log n)^2})$}. To detect a triangle, we can now, for each $v\in V_{1}$, partition the neighborhoods of $v$ in $V_{2}$ and $V_3$ into subsets of blocks, and test whether these subsets are connected by an edge using one query to the lookup table. In the positive case we have detected a triangle. For each vertex we perform \smash{$O(\frac{n^{2}}{(\log n)^2})$} queries and therefore the total running time is $O(\frac{n^{3}}{(\log n)^2})$. Interestingly, this technique can also be used to \emph{list} all triangles in the graph: Instead of storing a flag in the lookup table, we explicitly store a list of all edges between the sets $S$ and $S'$.

Vassilevska \cite{Vassilevska09Clique} demonstrated how the Four-Russians technique can be used for $k$-clique in a similar fashion: Partition $V_2, \dots, V_k$ into blocks and for every $(k-1)$-tuple of blocks we precompute whether there is a $(k-1)$-clique among these blocks. This results in time \smash{$\Order(\frac{n^k}{(\log n)^{k-1}})$} for $k$-clique detection.

\subsubsection{Other Approaches}
A line of work published in recent years obtained improvements beyond the Four-Russians bound. In the first such result, Bansal and Williams~\cite{BansalW12} combined regularity lemmas with the Four-Russians method. The basic intuition is that when looking at large (and therefore time-consuming) sets of neighbors of some vertex that lie within dense and random-like sets, there must exist an edge between the sets which reveals a triangle. They combined this intuition with the important observation that the Four-Russians method can be altered in a way that compounds sparseness in the graph (see \cref{lem:sparse-four-russians}), and were able to achieve a running time of $\Ologlog(\frac{n^{3}}{(\log(n))^{2.25}})$.

Chan~\cite{Chan15} improved this further using an entirely different approach. His main observation was that if there is a high-degree node, we can check its neighbors for the existence of an edge (such an edge exists if and only if the vertex is involved in a triangle), and recurse smartly on significantly smaller subproblems. This divide-and-conquer scheme effectively reduces Triangle Detection in dense graphs to Triangle Detection in sparse graphs, and leads to a running time of~\smash{$\Ologlog(\frac{n^3}{(\log n)^3})$}. Yu~\cite{Yu18} later refined this result and achieved a running time of~\smash{$\Ologlog(\frac{n^3}{(\log(n))^4})$}.
\section{Combinatorial Log-Shaves for \texorpdfstring{\boldmath$k$}{k}-Clique} \label{sec:k-clique}
In this section we provide our improved algorithmic reduction from $k$-clique to triangle detection (see \cref{thm:k-clique-to-triangle}). In our core we follow a divide-and-conquer approach for $k$-clique reminiscent to Chan's algorithm for triangle detection~\cite{Chan15} with a simple analysis. We start with the following observation:

\begin{observation}[Trivial Reduction from \boldmath$k$-Clique to $(k\mathord-1)$-Clique] \label{obs:k-clique-to-k-1-clique}
Let $k \geq 4$, let $f(n)$ be a nondecreasing function, and assume that there is a combinatorial $(k-1)$-clique detection algorithm running in time \smash{$\Order(n^{k-1} / f(n))$}. Then there is a combinatorial $k$-clique detection algorithm running in time
\begin{equation*}
    \Order\parens*{\sum_{v \in V_1} \frac{d_2(v)\cdot\dots\cdot d_k(v)}{f(\min\set{d_2(v),\dots,d_k(v)})}}.
\end{equation*}
\end{observation}
\begin{proof}
The algorithm is simple: For each vertex $v \in V_1$, we construct the subgraph $G_{v} = G[N_2(v) \cup \dots \cup N_k(v)]$ consisting of all neighbors of $v$ and test whether $G_{v}$ contains a $(k-1)$-clique. Let $n_{v} = d_2(v) + \dots + d_k(v)$ denote the number of vertices in $G_{v}$. Our intention is to use the efficient $(k-1)$\=/clique algorithm---however, simply running the algorithm in time $\Order(n_{v}^k / f(n_{v}))$ is possibly too slow. Instead, we partition each of the $k-1$ vertex parts in $G_{v}$ into blocks of size $d_{v} := \min\set{d_2(v), \dots, d_k(v)}$ (plus one final block of smaller size, respectively). Then, for each combination of $k-1$ blocks, we use the efficient $(k-1)$-clique detection algorithm. It is clear that the algorithm is correct, since we exhaustively test every tuple $(v_1, v_2, \dots, v_k)$. For the running time, note that testing whether $G_{v}$ contains a $k$-clique takes time
\begin{equation*}
    \ceil*{\frac{d_2(v)}{d_{v}}} \cdot \dots \cdot \ceil*{\frac{d_k(v)}{d_{v}}} \cdot \Order\parens*{\frac{(d_{v})^{k-1}}{f(d_{v})}} = \Order\parens*{\frac{d_2(v)\cdot\dots\cdot d_k(v)}{f(\min\set{d_2(v),\dots,d_k(v)})}},
\end{equation*}
and thus the total running time is indeed
\begin{equation*}
    \Order\parens*{\sum_{v \in V_1} \frac{d_2(v)\cdot\dots\cdot d_k(v)}{f(\min\set{d_2(v),\dots,d_k(v)})}}
\end{equation*}
(possibly after preprocessing the graph in time $\Order(n^2)$ to allow for constant-time edge queries. Note that this also covers the cost of constructing $G_v$ for every $v\in V_1$).
\end{proof}


Before moving to the formal proof of Theorem~\ref{thm:k-clique-to-triangle}, let us give a simplified high-level description of this algorithmic reduction in the specific case of 4-clique. For a given $4$-partite graph $(V_{1},V_{2},V_{3},V_{4})$, the core idea is the following: If the degrees in~$V_{1}$ tend to be small, i.e.\ if for every $v\in V_{1}$ we have $d_{2}(v)\cdot d_{3}(v)\cdot d_{4}(v)\leq\alpha\cdot|V_{2}|\cdot|V_{3}|\cdot|V_{4}|$ for some fraction $\alpha \approx \frac{1}{\log n}$, then we can apply \cref{obs:k-clique-to-k-1-clique}. Otherwise, there is a \emph{heavy} vertex $v\in V_{1}$ with $d_{2}(v)\cdot d_{3}(v)\cdot d_{4}(v)>\alpha\cdot|V_{2}|\cdot|V_{3}|\cdot|V_{4}|$. In this case, we will check every triplet of the form $(u,w,z)\in N_{2}(v)\times N_{3}(v)\times N_{4}(v)$. If any of these triplets form a triangle, we have detected a $4$-clique. Otherwise, we have learned that no triplet in $N_{2}(v)\times N_{3}(v)\times N_{4}(v)$ is part of a $4$-clique. We will therefore recurse in such a way that ensures we never test these triplets again and thereby make sufficient progress.

\begin{proof}
Assume that there is a combinatorial triangle detection algorithm which runs in time \smash{$\Order(n^3 (\log n)^a (\log\log n)^b)$}. We prove the claim by induction on $k$. The base case~($k = 3$) is immediate by the assumption there exists a triangle detection algorithm running in time \smash{$\Order(n^3 (\log n)^a (\log\log n)^b)$}.

For the inductive step, consider the following recursive algorithm to detect a $k$-clique in a given $k$-partite graph $(V_1, \dots, V_k, E)$. Let $D$ and $\alpha$ be parameters to be determined later and let $d$ be initialized to 0.

\medskip\noindent
\textsc{KCliqueRec}$(G=(V_1,\ldots,V_k,E), d)$:
\begin{enumerate}[topsep=\smallskipamount, itemsep=\smallskipamount]
    \item If $d=D$, meaning depth $D$ in the recursion is reached, perform exhaustive search. Return YES if a $k$-clique was detected, otherwise NO.
    
    \item Test whether there is some $v \in V_1$ with $d_2(v) \cdot \ldots \cdot d_k(v) \geq \alpha \cdot |V_2| \cdot \ldots \cdot |V_k|$. If such a vertex exists:
    \begin{enumerate}[topsep=\smallskipamount, itemsep=\smallskipamount, label=\alph*.]
        \item Test whether the subgraph $G_v$ induced by $N_2(v) \cup \dots \cup N_k(v)$ contains a $(k-1)$-clique by exhaustive search. If it does return YES since this means we've found a $k$-clique involving $v$.
        \item For $2 \leq i \leq k$, partition $V_i$ into $V_{i, 0} = V_i \setminus N_i(v)$ and \smash{$V_{i, 1} = V_i \cap N_i(v)$}.
        Recursively solve the $2^{k-1} - 1$ subproblems on $(V_1, V_{2, i_2}, \dots, V_{k, i_k})$ for $(i_2, \dots, i_k) \in \set{0, 1}^{k-1} \setminus \set{1^{k-1}}$, while incrementing the depth.\\
        In other words, for each $(i_2, \dots, i_k) \in \set{0, 1}^{k-1} \setminus \set{1^{k-1}}$, call \textsc{KCliqueRec}$(G[V_1 \cup V_{2, i_2} \cup \dots \cup V_{k, i_k}], d+1)$.
        \item If any of the calls returned YES, return YES. Otherwise, return NO.
    \end{enumerate}
    
    \item Solve the instance using \cref{obs:k-clique-to-k-1-clique}.
\end{enumerate}

\medskip\noindent\emph{Correctness.} As soon as the algorithm reaches recursion depth $D$, the algorithm will correctly detect a $k$-clique in step~1. In earlier levels of the recursion, the algorithm first attempts to find a vertex $v$ with~\makebox{$d_2(v) \cdot \ldots \cdot d_k(v) \geq \alpha \cdot |V_2| \cdot \ldots \cdot |V_k|$} in step~2. If this succeeds, we test whether $v$ is involved in a $k$-clique (and terminate in this case). Otherwise, we recurse on $(V_1, V_{2, i_2}, \dots, V_{k, i_k})$ for all combinations~\makebox{$(i_2, \dots, i_k) \in \set{0, 1}^{k-1} \setminus \set{1^{k-1}}$}. Note that we can indeed ignore the instance $(V_1, V_{2, 1}, \dots, V_{k, 1})$ knowing that $(V_{2, 1}, \dots, V_{k, 1})$ does not contain a $(k-1)$-clique. If the condition in step 2 is not satisfied, we instead correctly solve the instance by means of \cref{obs:k-clique-to-k-1-clique} (which reduces the problem to an instance of $(k-1)$-clique).

\medskip\noindent\emph{Running Time.}
Imagine a recursion tree in which every node corresponds to an execution of the algorithm; the root corresponds to the initial call and child nodes correspond to recursive calls. Thus, every node in the tree is either a leaf (indicating that this execution does not spawn recursive calls), or an internal node with fan-out exactly $2^{k-1} - 1$. The \emph{time} at a node is the running time of the respective call of the algorithm (ignoring the cost of further recursive calls). In other words, the \emph{time} at a node is the amount of local work performed in the corresponding call. To bound the total running time of the algorithm, we bound the total time across all nodes in the recursion tree.

We analyze the contributions of all steps individually. Let us introduce some notation first: At a node $x$ in the recursion tree, let $(V_1^x, \dots, V_k^x)$ denote the instance associated to the respective invocation. Similarly write $d_2^x(v), \dots, d_k^x(v)$.

\medskip\noindent\emph{Cost of Step 1.} 
Note that at any node $x$ at depth $D$ in the recursion tree, the time is $\Order(|V_1^x| \cdot \ldots \cdot |V_k^x|)$ since we solve the instance by exhaustive search. Next, observe that for any internal node $x$ in the recursion tree, we have that
\begin{gather*}
    |V_1^x| \cdot \ldots \cdot |V_k^x| = |V_1^x| \cdot \sum_{i_2, \dots, i_k \in \set{0, 1}^{k-1}} |V_{2, i_2}^x| \cdot \ldots \cdot |V_{k, i_k}^x| \\
    \qquad\geq |V_1^x| \cdot d_2^x(v) \cdot \ldots \cdot d_k^x(v) + \sum_{\text{$y$ child of $x$}} |V_1^y| \cdot \ldots \cdot |V_k^y| \\
    \qquad\geq \alpha \cdot |V_1^x| \cdot \ldots \cdot |V_k^x| + \sum_{\text{$y$ child of $x$}} |V_1^y| \cdot \ldots \cdot |V_k^y|,
\end{gather*}
and thus
\begin{equation*}
    \sum_{\text{$y$ child of $x$}} |V_1^y| \cdot \ldots \cdot |V_k^y| \leq (1 - \alpha) \cdot |V_1^x| \cdot \ldots \cdot |V_k^x|.
\end{equation*}
It follows by induction that at any depth $d \leq D$ in the recursion tree, we have that
\begin{equation*}
    \sum_{\text{$x$ at depth $d$}} |V_1^x| \cdot \ldots \cdot |V_k^x| \leq (1 - \alpha)^d n^k.
\end{equation*}
In particular, the total time of all nodes at depth $D$ is bounded by $\Order((1 - \alpha)^D n^k)$.



\medskip\noindent\emph{Cost of Step 2.} Note that the number of nodes in our recursion tree is at most~\makebox{$2^{k D}$} since the recursion tree has degree $\leq 2^k$ and the recursion depth is capped at~$D$. At each node, the time of step~2a is bounded by~\makebox{$\Order(n^{k-1})$} and the cost of step~2b is bounded by $\Order(n^2)$. Therefore, the total time of step~2 across all nodes is bounded by $\Order(2^{k D} n^{k-1})$.




\medskip\noindent\emph{Cost of Step 3.} By induction we have obtained a $(k-1)$-clique algorithm in time $\Order(n^{k-1} / f(n))$, where $f(n) = (\log n)^{-a + k - 4} (\log\log n)^{-b - (k - 4)}$. Therefore, by \cref{obs:k-clique-to-k-1-clique} the total time of step~3 across all nodes~$x$ in the recursion tree is
\begin{equation*}
    \Order\parens*{\sum_{\text{$x$ leaf}} \sum_{v \in V_1^x} \frac{d_2^x(v) \cdot \ldots \cdot d_k^x(v)}{f(\min\set{d_2^x(v), \dots, d_k^x(v)})}}.
\end{equation*}
To bound this quantity, we distinguish two subcases: A pair $(x, v)$ (where $x$ is a leaf in the recursion tree and $v \in V_1^u$) is called \emph{relevant} if $d_2^x(v), \dots, d_k^x(v) \geq \sqrt n$ (where $n$ is the initial number of nodes). On the one hand, it is easy to bound the total cost of all irrelevant pairs by
\begin{equation*}
    \Order\parens*{\sum_{\text{$(x, v)$ irrelevant}} \frac{d_2^x(v) \cdot \ldots \cdot d_k^x(v)}{f(\min\set{d_2^x(v), \dots, d_k^x(v)})}} \leq \Order(2^{kD} n^{k-1/2}),
\end{equation*}
since there are at most $2^{kD}$ nodes in the recursion tree. On the other hand, for any relevant pair~$(x, v)$, we have $\min\set{d_2^x(v), \dots, d_k^x(v)} \geq \sqrt n$. Moreover, since we reach step~3 of the algorithm we further know that $d_2^x(v) \cdot \ldots \cdot d_k^x(v) \leq \alpha |V_2^x| \cdot \ldots \cdot |V_k^x|$ (as otherwise the condition in step~2 had triggered). It follows that
\begin{gather*}
    \Order\parens*{\sum_{\text{$(x, v)$ relevant}} \frac{d_2^x(v) \cdot \ldots \cdot d_k^x(v)}{f(\min\set{d_2^x(v), \dots, d_k^x(v)})}} \\
    \qquad\leq \Order\parens*{\sum_{\text{$(x, v)$ relevant}} \frac{\alpha |V_2^x| \cdot \ldots \cdot |V_k^x|}{f(\sqrt n)}} \\
    \qquad\leq \Order\parens*{n^k \cdot \frac{\alpha}{f(\sqrt n)}}.
\end{gather*}

\medskip\noindent\emph{Choosing the Parameters.} Summing over all contributions computed before, the total running time is bounded by
\begin{equation*}
    \Order\parens*{n^k \cdot (1 - \alpha)^D + n^k \cdot \frac{\alpha}{f(\sqrt n)} + n^{k - 1/2} \cdot 2^{kD}}.
\end{equation*}
We pick $D = \log n / (4k)$ such that the latter term becomes $n^{k - 1/4}$. Next, we pick $\alpha = \log((-a + k) \log n) / D = \Theta((\log n)^{-1} \log\log n)$ such that the first term becomes
\begin{equation*}
    n^k \cdot (1 - \alpha)^D \leq n^k \cdot 2^{-\alpha D} \leq n^k (\log n)^{a-k}.
\end{equation*}
All in all, the total running time is dominated by the second term
\begin{gather*}
    n^k \cdot \frac{\alpha}{f(\sqrt n)} \leq \Order(n^k \cdot \alpha \cdot (\log n)^{a-(k-4)} (\log\log n)^{b+k-4}) \\
    \qquad\leq \Order(n^k (\log n)^{a-(k-3)} (\log\log n)^{b+k-3}),
\end{gather*}
which is as claimed.
\end{proof}

\section{Combinatorial Log-Shaves for Triangle Listing by Weak Regularity} \label{sec:triangle-listing}
In this section we reformulate and reanalyze Bansal and Williams' BMM algorithm~\cite{BansalW12} to support triangle \emph{listing}. Specifically, they provide a combinatorial algorithm for Boolean matrix multiplication in time~\smash{$\Ologlog(\frac{n^3}{(\log n)^{2.25}})$}, and our adaptation is as follows:

\thmtrianglelisting*


Note that we cannot achieve this running time by applying state-of-the-art black-box reductions from triangle listing to Boolean matrix multiplication (\cite{WilliamsW18}). Before we can present the algorithm we will need to define the notion of pseudoregular partitions in graphs and present a lemma by Frieze and Kannan~\cite{FriezeK99} known as the weak regularity lemma.

\subsection{Pseudoregularity}
Let $G=(V,E)$ be a graph and let $S,T\subseteq V$ be disjoint subsets of vertices. We define the \emph{density} of $(S,T)$ as
\begin{equation*}
    \delta(S,T) = \frac{e(S,T)}{|S|\cdot |T|},
\end{equation*}
where
\begin{equation*}
    e(S, T) = \abs{\set{(u,v)\in E:u\in S, v\in T}}.
\end{equation*}
Fix a partition $\mathcal{P} = (V_1, \dots, V_k)$ of the vertices $V$; we call the parts $V_1, \dots, V_k$ the \emph{pieces} of $\mathcal P$. For the sake convenience, we often employ the following shorthand notation: For a subset~\makebox{$S\subseteq V$} and $i\in [k]$, write~\makebox{$S_i := S\cap V_{i}$}. Moreover, we set $\delta_{i, j} := \delta(V_i, V_j)$. The partition~$\mathcal{P}$ is called \emph{$\varepsilon$-pseudoregular} if for every pair of disjoint subsets $S, T \subseteq V$ we have
\begin{equation*}
    \abs*{e(S, T) - \sum_{i, j \in [k]} \delta_{i, j} \cdot |S_i| \cdot |T_j|} \leq \varepsilon n^2.
\end{equation*}
We further call the partition $\mathcal P$ \emph{equitable} if the sizes of the vertex parts differ by at most~$1$ (i.e., $|\, |V_i| - |V_j|\,| \leq 1$ for all $i, j \in [k]$).

We can think about pseudoregularity in the following way: We compare the density between the pieces to the density of subsets of these pieces. In a random graph these densities are expected to be equal. If this was indeed the case, then \smash{$e(S, T) = \sum_{i,j \in [k]} \delta_{i,j}|S_{i}|\cdot|T_{j}|$}. Pseudoregularity requires that these quantities are only approximately equal, allowing an error term of $\varepsilon n^2$.

The following weak regularity lemma by Frieze and Kannan~\cite{FriezeK99} shows that we can always find a pseudoregular partition with relatively few parts, and that such a partition can be computed efficiently:

\begin{theorem}[{\cite[Theorem~2 and Section~5.1]{FriezeK99}}] \label{thm:pseudoregular}
Let $\varepsilon, \delta > 0$. Given an $n$\=/vertex graph, we can construct a $\varepsilon$-pseudoregular partition with at most~\smash{$2^{\Order(1/\varepsilon)}$} pieces in time \smash{$2^{O(1/\varepsilon^{2})} \cdot n^{2} \cdot \varepsilon^{-2}  \cdot \delta^{-3}$} by a randomized algorithm that succeeds with probability at least $1 - \delta$.
\end{theorem}

\subsection{Triangle Listing}
We now turn to the triangle listing algorithm à la Bansal-Williams~\cite{BansalW12}, starting with an informal overview. For a given tripartite graph $G = (V_1, V_2, V_3, E)$, we first compute an $\varepsilon$-pseudoregular partition of the bipartite graph~\makebox{$G[V_2 \cup V_3]$}. We then distinguish between two types of pieces---pieces with low density (less than $\sqrt{\varepsilon}$) and pieces with high density. Based on this we divide the instance into two triangle listing instances---$G_L$ which only includes edges connecting low density parts between $V_2$ and $V_3$ in $G$, and its complement $G_H$ consisting of edges connecting the high-density parts between $V_2$ and $V_3$. In the former case we can exploit the sparsity (by construction the total number of edges $G_L$ is at most $\sqrt{\varepsilon} n^2$). In the latter case, due to the pseudoregularity, there must be many triangles in $G_H$. We can thus charge the extra cost of computing with $G_H$ towards the output-size.

The key ingredient is the following \cref{lem:sparse-four-russians} that applies the Four-Russians method to \emph{sparse} graphs. Recall that the standard Four-Russians technique we presented in the introduction leads to a triangle listing algorithm in time $\Order(\frac{n^{3}}{\log^{2}(n)} + t)$.

\begin{lemma}[Sparse Four-Russians] \label{lem:sparse-four-russians}
There is an algorithm which lists up to~$t$ triangles in a given graph $(V_1, V_2, V_3, E)$ (with $n = \min\set{|V_1|, |V_2|, |V_3|}$) in time
\begin{equation*}
    \Ologlog\parens*{\frac{|V_1| \cdot |V_2| \cdot |V_3|}{(\log n)^{100}}+\sum_{v\in V_{1}} \frac{d_{2}(v)\cdot d_{3}(v)}{(\log n)^2}+ t}.
\end{equation*}
\end{lemma}
\begin{proof}
The idea is similar to the standard Four-Russians method: Let $s, \Delta$ be parameters to be determined. We partition $V_2$ into $g = \ceil{|V_2| / s}$ blocks $V_{2, 1}, \dots, V_{2, g}$ of size at most $s$, and similarly partition $V_3$ into $h = \ceil{|V_3| / s}$ blocks $V_{3, 1}, \dots, V_{3, h}$. Then consider the following three steps:
\begin{enumerate}[topsep=\smallskipamount, itemsep=\smallskipamount]
    \item We precompute, for each pair $i \in [g], j \in [h]$ and for each pair of subsets~\makebox{$S, T \subseteq [s]$} with size at most $|S|, |T| \leq \Delta$, the list of all edges in the induced graph on the vertex sets $V_{2, i} \cap S$ and $V_{3, j} \cap T$. Here, to represent a set~$S \subseteq [s]$ of size~$\Delta$ we fix a bit-representation consisting of at most \smash{$\ceil{\log \binom{s}{\Delta}}$} bits. We will later pick $s$ and $\Delta$ in such a way that \smash{$\binom{s}{\Delta} \leq n^{0.1}$}; note that any such set~$S$ can be represented by $\Order(1)$ machine words.
    \item For each vertex $v \in V_1$, we write $N_{2, i} = N_2(v) \cap V_{2, i}$ and~\makebox{$N_{3, j} = N_3(v) \cap V_{3, j}$}. Let us further set $d_{2, i} = |N_{2, i}|$ and $d_{3, j} = |N_{3, j}|$. We arbitrarily partition the set $N_{2, i}(v)$ into \smash{$\ceil{\frac{d_{2, i}}{\Delta}}$} subsets of size at most~$\Delta$; denote this partition by~$\mathcal S_i(v)$. We also compute a partition $\mathcal T_j(v)$ of~$N_{3, j}(v)$ with analogous properties. Here we represent each set $S \in \mathcal S(v), T \in \mathcal T(v)$ in its bit-representation.
    \item For each $v \in V_1$, for each $i \in [g], j \in [h]$, and for each~\makebox{$S \in \mathcal S_i(v), T \in \mathcal T_j(v)$}, list all edges $(v_2, v_3) \in S \times T$ (as precomputed in the first step). For each such edge, we report $(v, v_2, v_3)$ as a triangle. (If at some point we have listed~$t$ triangles, we stop the algorithm.)
\end{enumerate}

The correctness is straightforward: For each triangle $(v, v_2, v_3)$, we eventually test the indices $i \in [g], j \in [h]$ and sets $S \in S_i(v), T \in \mathcal T_j(v)$ where $v_2 \in S$ and $v_3 \in T$.

To analyze the running time, we choose the parameters $s = \floor{(\log n)^{100}}$ and $\Delta = \floor{\frac{\log n}{1000 \log\log n}}$. Note that for this choice we indeed have
\begin{equation*}
    \binom{s}{\Delta} \leq s^\Delta = (\log n)^{\frac{100 \log n}{1000 \log\log n}} = n^{0.1}.
\end{equation*}
Moreover, the precomputation produces a table of size $\Order(g^2 \cdot \binom{s}{\Delta}{}^2) = \Order(g^2 n^{0.2})$ and can be performed in time $\Order(g^2 \cdot \binom{s}{\Delta}{}^2 \cdot s^2) = \Order(n^{2.2})$. The second step takes time~\smash{$\Order(n^2 s) = \widetilde\Order(n^2)$}. Finally, in the third step we spend time
\begin{gather*}
    \Order\parens*{\sum_{v \in V_1} \sum_{i \in [g]} \sum_{j \in [h]} |\mathcal S_i(v)| \cdot |\mathcal T_j(v)| + t} \\
    \qquad=\Order\parens*{\sum_{v \in V_1} \parens*{\sum_{i \in [g]} \ceil*{\frac{d_{2, i}(v)}{\Delta}}} \cdot \parens*{\sum_{j \in [h]}\ceil*{\frac{d_{3, j}(v)}{\Delta}}} + t} \\
    \qquad=\Order\parens*{|V_1| \cdot |V_2| \cdot h + |V_1| \cdot g \cdot |V_3| + \sum_{v \in V_1} \frac{d_2(v) \cdot d_3(v)}{\Delta^2} + t} \\
    \qquad=\Order\parens*{\frac{|V_1| \cdot |V_2| \cdot |V_3|}{(\log n)^{100}} + \sum_{v \in V_1} \frac{d_2(v) \cdot d_3(v)}{\Delta^2} + t},
\end{gather*}
which is as claimed.
\end{proof}

\begin{corollary}[Sparse Four-Russians] \label{cor:sparse-four-russians}
There is an algorithm which lists up to~$t$ triangles in a given graph $(V_1, V_2, V_3, E)$ (with $n = \min\set{|V_1|, |V_2|, |V_3|}$) in time
\begin{equation*}
    \Ologlog\parens*{\frac{|V_1| \cdot |V_2| \cdot |V_3|}{(\log n)^{100}}+\frac{n \cdot e(V_2, V_3)}{(\log n)^2}+ t}.
\end{equation*}
\end{corollary}
\begin{proof}
Apply the previous lemma to the graph $(V_2, V_1, V_3, E)$. Since we can bound
\begin{equation*}
    \sum_{v \in V_2} \frac{d_1(v) \cdot d_3(v)}{(\log n)^2} \leq \frac{n e(V_2, V_3)}{(\log n)^2},
\end{equation*}
the claim follows.
\end{proof}

\begin{proof}[Proof of \cref{thm:triangle-listing}]
As a first step, we show how to list \emph{all} triangles in the given graph in time $\Ologlog(\frac{n^3}{(\log n)^{2.25}} + t_G)$, where $t_G$ is the total number of triangles in $G$. We will later remove this restriction.

Let $U = V_2 \cup V_3$ and consider the bipartite graph $G[U]$. Let $\epsilon > 0$ be a parameter to be determined later. We apply \cref{thm:pseudoregular} to $G[U]$ with parameters~$\epsilon$ and~\makebox{$\delta = n^{-0.1}$} to compute an $\epsilon$-pseudoregular partition $\mathcal P = (U_1, \dots, U_k)$ into at most \smash{$k = 2^{\Order(1/\epsilon^2)}$} pieces. Then, for each pair $i, j \in [k]$, construct the tripartite induced subgraph on~$V_1$,~\makebox{$V_{2, i} := V_2 \cap U_i$} and~\makebox{$V_{3, j} := V_3 \cap U_j$}, and use \cref{lem:sparse-four-russians} or \cref{cor:sparse-four-russians} to list all triangles in that graph (whichever is faster).

The correctness is immediate (and we will analyze the success probability later). It remains to argue that this algorithm runs in the claimed time budget. We need some notation: For a node $v \in V_1$ and $i \in [k]$, let $N_{2, i}(v) = N_1(v) \cap U_i$ and let $d_{2, i} = |N_{2, i}(v)|$ (and similarly for~$N_{3, i}, d_{3, i}$). Moreover, as before, let $\delta_{i, j} := \delta(U_i, U_j)$ denote the density between the pieces~$U_i$ and~$U_j$. Let us ignore the preprocessing for now. Then, writing
\begin{equation*}
    T_{i, j} = \min\set*{\sum_{v \in V_1} \frac{d_{2, i}(v) \cdot d_{3, j}(v)}{(\log n)^2}, \frac{n \cdot e(V_{2, i}, V_{3, j})}{(\log n)^2}},
\end{equation*}
the running time is bounded by
\begin{gather*}
    \Ologlog\parens*{\sum_{i, j \in [k]} \parens*{n^2 + \frac{|V_1| \cdot |V_{2, i}| \cdot |V_{3, j}|}{(\log n)^{100}} + T_{i, j}} + t_G} \\
    \qquad= \Ologlog\parens*{k^2 n^2 + \frac{n^3}{(\log n)^{100}} + \sum_{i, j \in [k]} T_{i, j} + t_G}.
\intertext{By setting~\smash{$\epsilon = \Theta(\frac{1}{\sqrt{\log n}})$} for an appropriately small hidden constant, this becomes}
    \qquad= \Ologlog\parens*{n^{2.5} + \frac{n^3}{(\log n)^{100}} + \sum_{i, j \in [k]} T_{i, j} + t_G}.
\end{gather*}
It remains to bound the sum $\sum_{i, j} T_{i, j}$. To this end we distinguish two cases: On the one hand, we have that
\begin{gather*}
    \sum_{\substack{i, j \in [k]\\\delta_{i, j} \leq \sqrt\epsilon}} T_{i, j}
    \leq \sum_{\substack{i, j \in [k]\\\delta_{i, j} \leq \sqrt\epsilon}} \frac{n \cdot e(V_{2, i}, V_{3, j})}{(\log n)^2}
    \leq \sum_{\substack{i, j \in [k]\\\delta_{i, j} \leq \sqrt\epsilon}} \frac{\sqrt\epsilon n \cdot |V_{2, i}| \cdot |V_{3, j}|}{(\log n)^2} \leq \frac{\sqrt\epsilon n^3}{(\log n)^2}.
\end{gather*}
On the other hand,
\begin{gather*}
    \sum_{\substack{i, j \in [k]\\\delta_{i, j} > \sqrt\epsilon}} T_{i, j} \leq \sum_{\substack{i, j \in [k]\\\delta_{i, j} > \sqrt\epsilon}} \sum_{v \in V_1} \frac{d_{2, i}(v) \cdot d_{3, j}(v)}{(\log n)^2} \\
    \qquad\leq \frac{1}{\sqrt\epsilon} \sum_{i, j \in [k]} \delta_{i, j} \sum_{v \in V_1} \frac{d_{2, i}(v) \cdot d_{3, j}(v)}{(\log n)^2} \\
    \qquad\leq \frac{1}{\sqrt\epsilon} \sum_{v \in V_1} \sum_{i, j \in [k]} \frac{\delta_{i, j} \cdot d_{2, i}(v) \cdot d_{3, j}(v)}{(\log n)^2} \\
\intertext{This is where finally the pseudoregularity becomes useful. Let $v \in V_1$ be arbitrary, let $S = N_2(v)$ and $T = N_3(v)$ (and write $S_i = N_{2, i}(v)$ and $T_j = N_{3, j}(v)$ as before). By the $\epsilon$-pseudoregularity we have $\sum_{i, j \in [k]} \delta_{i, j} |S_i| \cdot |T_i| \leq e(S, T) + \epsilon n^2$, and thus}
    \qquad\leq \frac{1}{\sqrt\epsilon} \sum_{v \in V_1} \frac{ \epsilon n^2 + e(N_2(v), N_3(v))}{(\log n)^2}
\intertext{Since clearly $\sum_{v \in V_1} e(N_2(v), N_3(v)) = t_G$, we finally obtain:}
    \qquad\leq \frac{\sqrt\epsilon n^3}{(\log n)^2} + \frac{t_G}{\sqrt\epsilon (\log n)^2}.
\end{gather*}
All in all, the remaining term in the running time can indeed be bounded by
\begin{equation*}
    \sum_{i, j \in [k]} T_{i, j} = \Order\parens*{\frac{\epsilon n^3}{(\log n)^2} + \frac{t_G}{\sqrt\epsilon (\log n)^2}} = \Order\parens*{\frac{n^3}{(\log n)^{2.25}} + t_G}.
\end{equation*}
Finally, recall that the preprocessing (to compute the pseudoregular partition) runs in time $2^{\Order(1/\epsilon^2)} \cdot n^2 \cdot \epsilon^{-2} \cdot \delta^{-3} = 2^{\Order(1/\epsilon^2)} \cdot \widetilde\Order(n^{2.3})$ (\cref{thm:pseudoregular}) which by our choice of $\epsilon$ becomes $\Order(n^{2.5})$, say.

Concerning the success probability, note that the only source of error in the algorithm is the computation of the pseudoregular partition which fails with probability at most $\delta = n^{-0.1}$. But even if computing this partition fails, we will nevertheless correctly report all triangles, and only the running time of our algorithm is affected. We can thus run our algorithm $1000$ times in parallel and stop as soon as the first copy reports an answer. The probability that all calls exceed their time budget is at most $(n^{-0.1})^{1000} = n^{-100}$ as claimed.

It remains to remove the assumption that we list all triangles, and instead only list triangles up to a given threshold $t$. To this end, we apply the following preprocessing to a given tripartite graph $(V_1, V_2, V_3, E)$: Split $V_1$ into $g := \Theta(\sqrt n)$ blocks $V_{1, 1}, \dots, V_{1, g}$ of size $\Theta(\sqrt n)$, say, and similarly for $V_2$ and $V_3$. For each triple $i, j, k \in [g]$ we list all triangles in the induced graph $G_{i, j, k} := G[V_{1, i} \cup V_{2, j} \cup V_{3, k}]$ in time
\begin{equation*}
    \Ologlog\parens*{\frac{n^{3/2}}{(\log n)^{2.25}} + t_{G_{i, j, k}}}
\end{equation*}
We stop as soon as we have listed $t$ triangles in total. Note that the total running time is thus bounded by $\Ologlog(\frac{n^3}{(\log n)^{2.25}} + t)$ plus~\smash{$\Ologlog(\max_{i, j, k} t_{G_{i, j, k}}) = \Ologlog(n^{3/2})$} for listing the surplus triangles. This overhead is negligible and the running time is as claimed.
\end{proof}

Note that from \cref{thm:triangle-listing} it easily follows that we can list \emph{all} triangles in time $\Order(n^3 / (\log n)^{2.25} + t_G)$, even in a black-box way. We simply apply \cref{thm:triangle-listing} with $t = n^3 / (\log n)^{2.25}$ and list up to $t$ triangles. As long as the graph contains at least $t$ triangles we double $t$ and repeat. The total running time is a geometric sum and thus bounded by $\Order(n^3 / (\log n)^{2.25} + t_G)$ as claimed.
\section{Combinatorial Log-Shaves for \texorpdfstring{\boldmath$k$}{k}-Hyperclique} \label{sec:k-hyper-clique}

We first provide an intuitive description of the algorithm in the simplest case $k=4$, $r=3$ (detecting a $4$-clique in a $3$-uniform hypergraph in faster than $\Order(n^4)$ time), see complete and general specification below. We are given a $4$-partite $3$-uniform graph $G = (V_1, V_2, V_3, V_4, E)$ with vertex sets of size $n$. For each $v\in V_{1}$, we can define a tri-partite graph $G_{v} = (V_{2},V_{3},V_{4},E')$ in which we draw an edge between two vertices if and only if they share a hyperedge with $v$ in $G$. It is easy to check that there is a 4-hyperclique in $G$ if and only if there are vertices $v_2, v_3, v_4$ that form a triangle in $G_v$ \emph{and} in $G$ (meaning they are a hyperedge in $G$). The naive search for such a triplet would take $O(n^3)$, and we present an algorithm that accelerates this search:

\begin{enumerate}[topsep=\smallskipamount, itemsep=\smallskipamount]
    \item Let $s=\sqrt{c \log n}$ for some small constant $c>0$, and partition $V_{2}$, $V_{3}$ and~$V_{4}$ each into $g=\lceil n/s\rceil$ blocks of size at most $s$. We let $V_{i,j}$ denote the $j$'th block in~$V_{i}$.
    
    \item For every combination $j_{2},j_{3},j_{4}\in\left[g\right]$:
    \begin{enumerate}[topsep=\smallskipamount, itemsep=\smallskipamount, label=\alph*.]
        \item Create a lookup table $T_{j_{2},j_{3},j_{4}}$ with an entry for every possible tripartite graph on the vertex sets $V_{2,j_{2}},V_{3,j_{3}},V_{4,j_{4}}$ (there  are \smash{$2^{s^{2}}= n^{c}$} such graphs).
        \item For every entry corresponding to a graph $G'$ store whether $G'$ has a triangle that is a hyperedge in $G$.
    \end{enumerate}
\end{enumerate}

Note that this preprocessing is fast: We construct $\frac{n^{3}}{s^{3}}$ tables, each consisting of $n^{c}$ entries, and each entry takes $O(s^{3})$ time to determine. So, the total preprocessing time is $O(n^{3+c})$.
Given these tables we can now search for a 4\=/clique more efficiently: For each $v\in V_{1}$ we break $G_{v}$ into triples of blocks as before, and query $T_{j_{2},j_{3},j_{4}}$ for the graphs $G_{v}[V_{2,j_{2}} \cup V_{3,j_{3}} \cup V_{4,j_{4}}]$, for all $j_2, j_3, j_4$. If one the answers is positive we have found a hyperclique.
Assuming every query is performed in constant time, the running time is determined by the number of queries which is
\begin{equation*}
    O\parens*{n \cdot \frac{n^{3}}{s^{3}}} = O\parens*{\frac{n^{4}}{(\log n)^{1.5}}}.
\end{equation*}

All that is left now is to justify the assumption that every query is performed in constant time. The main question is given $v\in V_{1}$ and a combination of blocks $V_{2,j_{2}},V_{3,j_{3}},V_{4,j_{4}}$, how can we determine the key corresponding to~\smash{$G_{v}[V_{2,j_{2}},V_{3,j_{3}},V_{4,j_{4}}]$} in $T_{j_{2},j_{3},j_{4}}$ in constant time? For this purpose, we define in the proof a \emph{compact representation} of tripartite graphs (on vertex sets of size~$s$) used to index the tables $T_{j_2, j_3, j_4}$. This compact representation is chosen in such a way which allows to efficiently precompute the compact representations of all such graphs ~$G_v[V_{2,j_{2}},V_{3,j_{3}},V_{4,j_{4}}]$.

\paragraph{Complete description}We now turn to describe the algorithm in full generality and detail. We will in fact extend it to \emph{list} up to $t$ hypercliques in output-sensitive time (see \cref{thm:k-hyperclique}). Throughout, $G = (V_1, \dots, V_k, E)$ is a $k$-partite $r$-uniform hypergraph. For a vertex $v \in V_1$, we define the \emph{adjacency subgraph of $v$}, denoted by $G_v$, to be the $(k-1)$-partite $(r-1)$-uniform hypergraph with vertex sets $V_2, \dots, V_k$ which has a hyperedge $\set{u_1, \dots, u_{r-1}}$ if and only there is a hyperedge $\set{v, u_1, \dots, u_{r-1}} \in E$ in $G$.

\begin{observation} \label{obs:hyperclique-subgraph}
Let $3 \leq r < k$, and let $(V_1, \dots, V_k, E)$ be a $k$-partite $r$-uniform hypergraph. Then, for every tuple $v_1 \in V_1, \dots, v_k \in V_k$, $(v_1, \dots, v_k)$ is a $k$\=/hyperclique in $G$ if and only if $(v_2, \dots, v_k)$ is a $(k-1)$-hyperclique in $G_{v_1}$ and is also a $(k-1)$-hyperclique (or a hyperedge if $r=k-1$) in $G$.
\end{observation}

Let us briefly reflect on this observation. To find a k-clique in a graph we can search through the neighborhoods of each vertex for a $(k-1)$-clique. The Four-Russians method involves preprocessing the graph in order to perform this search more efficiently. \cref{obs:hyperclique-subgraph} gives rise to an analogous process to find $k$-hypercliques in hypergraphs: Searching through the induced adjacency graph of each vertex for a $(k-1)$-hyperclique \emph{that is also a clique in $G$}. We present an adaptation of the Four-Russians method to perform this analogous search more efficiently.

\thmkhyperclique*
\begin{proof}
Let $s$ be a parameter to be determined. As a first step, we partition the vertex sets $V_2, \dots, V_k$ into $g := \ceil{n / s}$ \emph{blocks} $V_i = V_{i, 1} \cup \dots \cup V_{i, g}$ of size at most~$s$. For $j = (j_2, \dots, j_k) \in [g]^{k-1}$, we let
\begin{equation*}
    V^j = V_{2, j_2} \cup \dots \cup V_{k, j_k},
\end{equation*}
meaning $V^j$ denotes some combination of blocks across the vertex sets, and we let $G^j$ denote the $(k-1)$-partite subgraph of $G$ induced by $V^j$. Similarly, for $v \in V_1$, we let $G_v^j$ denote the subgraph of $G_v$ induced by $V^j$. It follows from the preceding discussion that there is a $k$-hyperclique in $G$ if and only if there is some $v \in V_1$ and some $j \in [g]^{k-1}$ such that $G^j$ and $G_v^j$ share a common $(k-1)$-hyperclique. To detect whether $G^j$ and $G_v^j$ share a common hyperclique, we now present an efficient algorithm.

\medskip\noindent\emph{Compressed Representation.} A critical feature of the Four-Russians technique is that we need to \emph{compactly represent} small graphs. Specifically, let $j = (j_2, \dots, j_k)$ and consider a $(k-1)$-partite $(r-1)$-regular hypergraph $(V^j, H)$ on the vertex set~$V^j$. We will fix a description of this graph as a sequence of bits. First, partition the hyperedges $H$ into parts $H_I$ for every set~\smash{$I = \set{i_1 < \dots < i_{r-1}} \subseteq \set{2, \dots, k}$}, such that the part $H_I$ contains exactly the hyperedges involving the vertex parts~\smash{$V_{i_1, j_{i_1}}, \dots, V_{i_{r-1}, j_{i_{r-1}}}$}. We define the compact representation of $H_I$ as the bit-string of length~\smash{$s^{r-1}$} obtained by listing (in an arbitrary but fixed order) all tuples~$(v_{i_1}, \dots, v_{i_{r-1}}) \in V_{i_1, j_{i_1}} \times \dots \times V_{i_{r-1}, j_{i_{r-1}}}$, where we indicate present edges by $1$ and missing edges by $0$. Finally, the compact representation of $H$ is defined as the concatenation of the compact representations of the parts $H_I$ (in an arbitrary but consistent order).

Note that this compact representation has length exactly \makebox{$L = \binom{k - 1}{r - 1} \cdot s^{r-1}$}. By choosing
\begin{equation*}
    s := \parens*{\frac{\log n}{2 \binom{k-1}{r-1}}}^{\frac{1}{r-1}} = \Theta((\log n)^{\frac{1}{r-1}}),
\end{equation*}
the length becomes at most $L \leq \frac12 \log n$. In particular, we can store the compact representation of any graph $H$ as before in $\Order(1)$ machine words.

In the next two claims we show that we can preprocess \emph{all} graphs~$H$. 

\begin{claim} \label{thm:k-hyperclique:clm:indexing}
In time $\Order(n^{k-1} 2^L s^k)$ we can compute a data structure that, given $j \in [g]^{k-1}$ and an $(r-1)$-uniform hypergraph on $V^j$ (represented compactly), tests in $\Order(1)$ time whether there are vertices $v_2, \dots, v_k \in V^j$ such that
\begin{itemize}
    \item $(v_2, \dots, v_k)$ is a $(k-1)$-hyperclique in $G^j$, and
    \item $(v_2, \dots, v_k)$ is a $(k-1)$-hyperclique in $H$.
\end{itemize}
Moreover, we can enumerate all such tuples $(v_2, \dots, v_k)$ with constant delay.
\end{claim}
\begin{proof}[Proof of \cref{thm:k-hyperclique:clm:indexing}]
The data structure consists of an array of length $g^{k-1} \cdot 2^L$. Each position is associated to a pair $(j, H)$, where $j \in [g]^{k-1}$ and where $H$ is a $(r-1)$-uniform hypergraph on $V^j$ (that can be represented compactly in $L$ bits). To fill the array, we enumerate all positions $(j, H)$ and test by exhaustive search whether there is a tuple $(v_2, \dots, v_k) \in V^j$ that forms a hyperclique in $G^j$ and in $H$. Each such test requires time $\Order(s^k)$, therefore the total running time is indeed $\Order(n^{k-1} 2^L s^k)$.

Each query testing whether there is a common hyperclique in $G^j$ and $H$ indeed runs in constant time (implemented by one lookup operation). For the enumeration, we separately store for each entry $(j, H)$ in the array a list of all common hypercliques $(v_2, \dots, v_k)$ of $G^j$ and $H$.
\end{proof}

\begin{claim} \label{thm:k-hyperclique:clm:compressed}
In time $\Order(n^k / s^{k-1})$ we can compute for all $j \in [g]^{k-1}$ and all $v \in V_1$ the graphs $G_v^j$ in compressed representation.
\end{claim}
\begin{proof}[Proof of \cref{thm:k-hyperclique:clm:compressed}]
Recall that the compact representation of $G_v^j$ is the concatenation of the compact representations of $(G_v^j)_I$ for all sets $I = \set{i_1 < \dots \nobreak<\nobreak i_{r-1}} \subseteq \set{2, \dots, k}$. Thus, in a first step we will precompute the compact representations of $(G_v^j)_I$. Note that this representation only depends on $j_I := (j_{i_1}, \dots, j_{i_{r-1}})$ (and not on the remaining $j$-indices). We can precompute all compressed representations determined by $v \in V_1$ and $j_I \in [g]^{r-1}$ in time $\Order(n \cdot g^{r-1}) = \Order(n^r)$. After this precomputation, we can assemble the compressed representation of any graph $G_v^j$ in constant time $\Order(\binom{k-1}{r-1}) = \Order(1)$. Therefore, the total time is~\smash{$\Order(n^r + n \cdot g^{k-1} = n^k / s^{k-1})$}.
\end{proof}

Given the previous \cref{thm:k-hyperclique:clm:indexing,thm:k-hyperclique:clm:compressed}, the proof of the theorem is almost complete. We first use \cref{thm:k-hyperclique:clm:compressed} to prepare the compressed representations of all graphs $G_v^j$ (for $v \in V_1$ and $j \in [g]^{k-1}$). Using \cref{thm:k-hyperclique:clm:indexing} we prepare a data structure that decide in constant time for each $G_v^j$ whether it shares a $(k-1)$-hyperclique with $G^j$. Whenever this test succeeds, we enumerate all common hypercliques $(v_2, \dots, v_k)$ in $G^j$ and $G_v^j$ with constant delay, and report $(v, v_2, \dots, v_k)$. If at some point during the execution we have listed $t$ $k$-hypercliques, we interrupt the algorithm. As mentioned before, the correctness follows by \cref{obs:hyperclique-subgraph}. The running time is
\begin{multline*}
    \Order(n^{k-1} 2^L s^k) + \Order\parens*{\frac{n^k}{s^{k-1}}} + \Order(t) = \Olog(n^{k-1/2}) + \Order\parens*{\frac{n^k}{s^{k-1}} + t} \\ = \Order\parens*{\frac{n^k}{(\log n)^{\frac{k-1}{r-1}}} + t}
\end{multline*}
as claimed.
\end{proof}




\begin{corollary}[Faster \boldmath$k$-Hyperclique Detection] \label{cor:k-hyperclique}
There is an algorithm for detecting $k$-hypercliques in $r$-uniform hypergraphs running in time
\begin{equation*}
    \Order\parens*{\frac{n^{k}}{(\log n)^{\frac{k-1}{r-1}}}}
\end{equation*}
(assuming a word RAM model with word size $w = \Omega(\log n)$).    
\end{corollary}
\begin{proof}
Call \cref{thm:k-hyperclique} with $t = 1$.
\end{proof}


\section*{Acknowledgements}
We would like to thank Oded Goldreich and Nathan Wallheimer for discussions on combinatorial BMM algorithms.

\bibliographystyle{plainurl}
\bibliography{references}

\begin{thebibliography}{10}

\bibitem{abboud2017fine}
Amir Abboud, Arturs Backurs, Karl Bringmann, and Marvin K{\"u}nnemann.
\newblock Fine-grained complexity of analyzing compressed data: Quantifying improvements over decompress-and-solve.
\newblock In {\em 2017 IEEE 58th Annual Symposium on Foundations of Computer Science (FOCS)}, pages 192--203. IEEE, 2017.

\bibitem{ABV15}
Amir Abboud, Arturs Backurs, and Virginia~Vassilevska Williams.
\newblock If the current clique algorithms are optimal, so is valiant's parser.
\newblock {\em SIAM Journal on Computing}, 47(6):2527--2555, 2018.

\bibitem{abboud2018more}
Amir Abboud, Karl Bringmann, Holger Dell, and Jesper Nederlof.
\newblock More consequences of falsifying seth and the orthogonal vectors conjecture.
\newblock In {\em Proceedings of the 50th Annual ACM SIGACT Symposium on Theory of Computing}, pages 253--266, 2018.

\bibitem{AFKLM23}
Amir Abboud, Nick Fischer, Zander Kelley, Shachar Lovett, and Raghu Meka.
\newblock New graph decompositions and combinatorial boolean matrix multiplication algorithms.
\newblock {\em CoRR}, abs/2311.09095, 2023.
\newblock URL: \url{https://doi.org/10.48550/arXiv.2311.09095}, \href {https://arxiv.org/abs/2311.09095} {\path{arXiv:2311.09095}}, \href {https://doi.org/10.48550/ARXIV.2311.09095} {\path{doi:10.48550/ARXIV.2311.09095}}.

\bibitem{AbboudGIKPTUW19}
Amir Abboud, Loukas Georgiadis, Giuseppe~F. Italiano, Robert Krauthgamer, Nikos Parotsidis, Ohad Trabelsi, Przemyslaw Uznanski, and Daniel Wolleb{-}Graf.
\newblock Faster algorithms for all-pairs bounded min-cuts.
\newblock In Christel Baier, Ioannis Chatzigiannakis, Paola Flocchini, and Stefano Leonardi, editors, {\em 46th International Colloquium on Automata, Languages, and Programming, {ICALP} 2019, July 9-12, 2019, Patras, Greece}, volume 132 of {\em LIPIcs}, pages 7:1--7:15. Schloss Dagstuhl - Leibniz-Zentrum f{\"{u}}r Informatik, 2019.
\newblock \href {https://doi.org/10.4230/LIPIcs.ICALP.2019.7} {\path{doi:10.4230/LIPIcs.ICALP.2019.7}}.

\bibitem{AbboudK0PST22}
Amir Abboud, Robert Krauthgamer, Jason Li, Debmalya Panigrahi, Thatchaphol Saranurak, and Ohad Trabelsi.
\newblock Breaking the cubic barrier for all-pairs max-flow: Gomory-hu tree in nearly quadratic time.
\newblock In {\em 63rd {IEEE} Annual Symposium on Foundations of Computer Science, {FOCS} 2022, Denver, CO, USA, October 31 - November 3, 2022}, pages 884--895. {IEEE}, 2022.
\newblock \href {https://doi.org/10.1109/FOCS54457.2022.00088} {\path{doi:10.1109/FOCS54457.2022.00088}}.

\bibitem{AbboudKT21}
Amir Abboud, Robert Krauthgamer, and Ohad Trabelsi.
\newblock Subcubic algorithms for gomory-hu tree in unweighted graphs.
\newblock In Samir Khuller and Virginia~Vassilevska Williams, editors, {\em {STOC} '21: 53rd Annual {ACM} {SIGACT} Symposium on Theory of Computing, Virtual Event, Italy, June 21-25, 2021}, pages 1725--1737. {ACM}, 2021.
\newblock \href {https://doi.org/10.1145/3406325.3451073} {\path{doi:10.1145/3406325.3451073}}.

\bibitem{AbboudW14}
Amir Abboud and Virginia~Vassilevska Williams.
\newblock Popular conjectures imply strong lower bounds for dynamic problems.
\newblock In {\em 55th Annual {IEEE} Symposium on Foundations of Computer Science ({FOCS} 2014)}, pages 434--443. {IEEE} Computer Society, 2014.
\newblock \href {https://doi.org/10.1109/FOCS.2014.53} {\path{doi:10.1109/FOCS.2014.53}}.

\bibitem{abboud2014consequences}
Amir Abboud, Virginia~Vassilevska Williams, and Oren Weimann.
\newblock Consequences of faster alignment of sequences.
\newblock In {\em Automata, Languages, and Programming: 41st International Colloquium, ICALP 2014, Copenhagen, Denmark, July 8-11, 2014, Proceedings, Part I 41}, pages 39--51. Springer, 2014.

\bibitem{AingworthCM96}
Donald Aingworth, Chandra Chekuri, and Rajeev Motwani.
\newblock Fast estimation of diameter and shortest paths (without matrix multiplication).
\newblock In {\'{E}}va Tardos, editor, {\em Proceedings of the Seventh Annual {ACM-SIAM} Symposium on Discrete Algorithms, 28-30 January 1996, Atlanta, Georgia, {USA}}, pages 547--553. {ACM/SIAM}, 1996.
\newblock URL: \url{http://dl.acm.org/citation.cfm?id=313852.314117}.

\bibitem{angluin1976four}
Dana Angluin.
\newblock The four russians' algorithm for boolean matrix multiplication is optimal in its class.
\newblock {\em ACM SIGACT News}, 8(1):29--33, 1976.

\bibitem{FourRussians70}
Vladimir~L'vovich Arlazarov, Yefim~A Dinitz, MA~Kronrod, and IgorAleksandrovich Faradzhev.
\newblock On economical construction of the transitive closure of an oriented graph.
\newblock {\em Doklady Akademii Nauk}, 194(3):487--488, 1970.

\bibitem{BackursDT16}
Arturs Backurs, Nishanth Dikkala, and Christos Tzamos.
\newblock Tight hardness results for maximum weight rectangles.
\newblock In Ioannis Chatzigiannakis, Michael Mitzenmacher, Yuval Rabani, and Davide Sangiorgi, editors, {\em 43rd International Colloquium on Automata, Languages, and Programming, {ICALP} 2016, July 11-15, 2016, Rome, Italy}, volume~55 of {\em LIPIcs}, pages 81:1--81:13. Schloss Dagstuhl - Leibniz-Zentrum f{\"{u}}r Informatik, 2016.
\newblock \href {https://doi.org/10.4230/LIPIcs.ICALP.2016.81} {\path{doi:10.4230/LIPIcs.ICALP.2016.81}}.

\bibitem{BackursT17}
Arturs Backurs and Christos Tzamos.
\newblock Improving viterbi is hard: Better runtimes imply faster clique algorithms.
\newblock In Doina Precup and Yee~Whye Teh, editors, {\em Proceedings of the 34th International Conference on Machine Learning, {ICML} 2017, Sydney, NSW, Australia, 6-11 August 2017}, volume~70 of {\em Proceedings of Machine Learning Research}, pages 311--321. {PMLR}, 2017.
\newblock URL: \url{http://proceedings.mlr.press/v70/backurs17a.html}.

\bibitem{BansalW12}
Nikhil Bansal and Ryan Williams.
\newblock Regularity lemmas and combinatorial algorithms.
\newblock {\em Theory Comput.}, 8(1):69--94, 2012.
\newblock \href {https://doi.org/10.4086/toc.2012.v008a004} {\path{doi:10.4086/toc.2012.v008a004}}.

\bibitem{BaranDP08}
Ilya Baran, Erik~D. Demaine, and Mihai Patrascu.
\newblock Subquadratic algorithms for {3SUM}.
\newblock {\em Algorithmica}, 50(4):584--596, 2008.
\newblock \href {https://doi.org/10.1007/s00453-007-9036-3} {\path{doi:10.1007/s00453-007-9036-3}}.

\bibitem{bergamaschi2021new}
Thiago Bergamaschi, Monika Henzinger, Maximilian~Probst Gutenberg, Virginia~Vassilevska Williams, and Nicole Wein.
\newblock New techniques and fine-grained hardness for dynamic near-additive spanners.
\newblock In {\em Proceedings of the 2021 ACM-SIAM Symposium on Discrete Algorithms (SODA)}, pages 1836--1855. SIAM, 2021.

\bibitem{bernardini2019even}
Giulia Bernardini, Pawel Gawrychowski, Nadia Pisanti, Solon~P Pissis, Giovanna Rosone, et~al.
\newblock Even faster elastic-degenerate string matching via fast matrix multiplication.
\newblock {\em LEIBNIZ INTERNATIONAL PROCEEDINGS IN INFORMATICS}, 132:1--15, 2019.

\bibitem{DBLP:conf/icalp/BjorklundPWZ14}
Andreas Bj{\"{o}}rklund, Rasmus Pagh, Virginia~Vassilevska Williams, and Uri Zwick.
\newblock Listing triangles.
\newblock In Javier Esparza, Pierre Fraigniaud, Thore Husfeldt, and Elias Koutsoupias, editors, {\em Automata, Languages, and Programming - 41st International Colloquium, {ICALP} 2014, Copenhagen, Denmark, July 8-11, 2014, Proceedings, Part {I}}, volume 8572 of {\em Lecture Notes in Computer Science}, pages 223--234. Springer, 2014.
\newblock \href {https://doi.org/10.1007/978-3-662-43948-7\_19} {\path{doi:10.1007/978-3-662-43948-7\_19}}.

\bibitem{bringmann2019fine}
Karl Bringmann, Nick Fischer, and Marvin K{\"u}nnemann.
\newblock A fine-grained analogue of schaefer's theorem in p: Dichotomy of exists\^{} k-forall-quantified first-order graph properties.
\newblock In {\em 34th Computational Complexity Conference (CCC 2019)}. Schloss Dagstuhl-Leibniz-Zentrum f{\"u}r Informatik, 2019.

\bibitem{bringmann2020tree}
Karl Bringmann, Pawe{\l} Gawrychowski, Shay Mozes, and Oren Weimann.
\newblock Tree edit distance cannot be computed in strongly subcubic time (unless apsp can).
\newblock {\em ACM Transactions on Algorithms (TALG)}, 16(4):1--22, 2020.

\bibitem{bringmann2019truly}
Karl Bringmann, Fabrizio Grandoni, Barna Saha, and Virginia~Vassilevska Williams.
\newblock Truly subcubic algorithms for language edit distance and rna folding via fast bounded-difference min-plus product.
\newblock {\em SIAM Journal on Computing}, 48(2):481--512, 2019.

\bibitem{bringmann2017dichotomy}
Karl Bringmann, Allan Gr{\o}nlund, and Kasper~Green Larsen.
\newblock A dichotomy for regular expression membership testing.
\newblock In {\em 2017 IEEE 58th Annual Symposium on Foundations of Computer Science (FOCS)}, pages 307--318. IEEE, 2017.

\bibitem{bringmann2018clique}
Karl Bringmann and Philip Wellnitz.
\newblock Clique-based lower bounds for parsing tree-adjoining grammars.
\newblock {\em arXiv preprint arXiv:1803.00804}, 2018.

\bibitem{CK19}
Nofar Carmeli and Markus Kr{\"{o}}ll.
\newblock On the enumeration complexity of unions of conjunctive queries.
\newblock In Dan Suciu, Sebastian Skritek, and Christoph Koch, editors, {\em Proceedings of the 38th {ACM} {SIGMOD-SIGACT-SIGAI} Symposium on Principles of Database Systems, {PODS} 2019, Amsterdam, The Netherlands, June 30 - July 5, 2019}, pages 134--148. {ACM}, 2019.
\newblock \href {https://doi.org/10.1145/3294052.3319700} {\path{doi:10.1145/3294052.3319700}}.

\bibitem{casel2021fine}
Katrin Casel and Markus~L Schmid.
\newblock Fine-grained complexity of regular path queries.
\newblock {\em arXiv preprint arXiv:2101.01945}, 2021.

\bibitem{Chan08}
Timothy~M Chan.
\newblock A (slightly) faster algorithm for klee's measure problem.
\newblock In {\em Proceedings of the twenty-fourth annual symposium on Computational geometry}, pages 94--100, 2008.

\bibitem{Chan15}
Timothy~M. Chan.
\newblock Speeding up the four russians algorithm by about one more logarithmic factor.
\newblock In Piotr Indyk, editor, {\em Proceedings of the Twenty-Sixth Annual {ACM-SIAM} Symposium on Discrete Algorithms, {SODA} 2015, San Diego, CA, USA, January 4-6, 2015}, pages 212--217. {SIAM}, 2015.
\newblock \href {https://doi.org/10.1137/1.9781611973730.16} {\path{doi:10.1137/1.9781611973730.16}}.

\bibitem{chan2020range}
Timothy~M Chan, Saladi Rahul, and Jie Xue.
\newblock Range closest-pair search in higher dimensions.
\newblock {\em Computational Geometry}, 91:101669, 2020.

\bibitem{Chang16}
Yi{-}Jun Chang.
\newblock Hardness of {RNA} folding problem with four symbols.
\newblock In Roberto Grossi and Moshe Lewenstein, editors, {\em 27th Annual Symposium on Combinatorial Pattern Matching, {CPM} 2016, June 27-29, 2016, Tel Aviv, Israel}, volume~54 of {\em LIPIcs}, pages 13:1--13:12. Schloss Dagstuhl - Leibniz-Zentrum f{\"{u}}r Informatik, 2016.
\newblock \href {https://doi.org/10.4230/LIPIcs.CPM.2016.13} {\path{doi:10.4230/LIPIcs.CPM.2016.13}}.

\bibitem{clifford2018upper}
Raphael Clifford, Allan Gr{\o}nlund, Kasper~Green Larsen, and Tatiana Starikovskaya.
\newblock Upper and lower bounds for dynamic data structures on strings.
\newblock {\em arXiv preprint arXiv:1802.06545}, 2018.

\bibitem{virginia_listing_cliques}
Mina Dalirrooyfard, Surya Mathialagan, Virginia~Vassilevska Williams, and Yinzhan Xu.
\newblock Listing cliques from smaller cliques.
\newblock {\em CoRR}, abs/2307.15871, 2023.
\newblock \href {https://arxiv.org/abs/2307.15871} {\path{arXiv:2307.15871}}, \href {https://doi.org/10.48550/arXiv.2307.15871} {\path{doi:10.48550/arXiv.2307.15871}}.

\bibitem{dalirrooyfard2019graph}
Mina Dalirrooyfard, Thuy~Duong Vuong, and Virginia~Vassilevska Williams.
\newblock Graph pattern detection: Hardness for all induced patterns and faster non-induced cycles.
\newblock In {\em Proceedings of the 51st Annual ACM SIGACT Symposium on Theory of Computing}, pages 1167--1178, 2019.

\bibitem{DasKS18}
Debarati Das, Michal Kouck{\'{y}}, and Michael~E. Saks.
\newblock Lower bounds for combinatorial algorithms for boolean matrix multiplication.
\newblock In Rolf Niedermeier and Brigitte Vall{\'{e}}e, editors, {\em 35th Symposium on Theoretical Aspects of Computer Science, {STACS} 2018, February 28 to March 3, 2018, Caen, France}, volume~96 of {\em LIPIcs}, pages 23:1--23:14. Schloss Dagstuhl - Leibniz-Zentrum f{\"{u}}r Informatik, 2018.
\newblock \href {https://doi.org/10.4230/LIPIcs.STACS.2018.23} {\path{doi:10.4230/LIPIcs.STACS.2018.23}}.

\bibitem{DuanWZ23}
Ran Duan, Hongxun Wu, and Renfei Zhou.
\newblock Faster matrix multiplication via asymmetric hashing.
\newblock In {\em 64th {IEEE} Annual Symposium on Foundations of Computer Science ({FOCS} 2023)}. {IEEE} Computer Society, 2023.
\newblock To appear.
\newblock URL: \url{https://doi.org/10.48550/arXiv.2210.10173}.

\bibitem{EG04}
Friedrich Eisenbrand and Fabrizio Grandoni.
\newblock On the complexity of fixed parameter clique and dominating set.
\newblock {\em Theor. Comput. Sci.}, 326(1-3):57--67, 2004.
\newblock \href {https://doi.org/10.1016/j.tcs.2004.05.009} {\path{doi:10.1016/j.tcs.2004.05.009}}.

\bibitem{DBLP:FoxRemove}
Jacob Fox.
\newblock A new proof of the graph removal lemma.
\newblock {\em CoRR}, abs/1006.1300, 2010.
\newblock URL: \url{http://arxiv.org/abs/1006.1300}, \href {https://arxiv.org/abs/1006.1300} {\path{arXiv:1006.1300}}.

\bibitem{FriezeK99}
Alan~M. Frieze and Ravi Kannan.
\newblock Quick approximation to matrices and applications.
\newblock {\em Comb.}, 19(2):175--220, 1999.
\newblock \href {https://doi.org/10.1007/s004930050052} {\path{doi:10.1007/s004930050052}}.

\bibitem{HenzingerKNS15}
Monika Henzinger, Sebastian Krinninger, Danupon Nanongkai, and Thatchaphol Saranurak.
\newblock Unifying and strengthening hardness for dynamic problems via the online matrix-vector multiplication conjecture.
\newblock In Rocco~A. Servedio and Ronitt Rubinfeld, editors, {\em Proceedings of the Forty-Seventh Annual {ACM} on Symposium on Theory of Computing, {STOC} 2015, Portland, OR, USA, June 14-17, 2015}, pages 21--30. {ACM}, 2015.
\newblock \href {https://doi.org/10.1145/2746539.2746609} {\path{doi:10.1145/2746539.2746609}}.

\bibitem{JX22}
Ce~Jin and Yinzhan Xu.
\newblock Tight dynamic problem lower bounds from generalized bmm and omv.
\newblock In {\em Proceedings of the 54th Annual ACM SIGACT Symposium on Theory of Computing}, pages 1515--1528, 2022.

\bibitem{KopelowitzPP16}
Tsvi Kopelowitz, Seth Pettie, and Ely Porat.
\newblock Higher lower bounds from the {3SUM} conjecture.
\newblock In Robert Krauthgamer, editor, {\em 27th Annual {ACM-SIAM} Symposium on Discrete Algorithms ({SODA} 2016)}, pages 1272--1287. {SIAM}, 2016.
\newblock \href {https://doi.org/10.1137/1.9781611974331.ch89} {\path{doi:10.1137/1.9781611974331.ch89}}.

\bibitem{KopelowitzP18}
Tsvi Kopelowitz and Ely Porat.
\newblock A simple algorithm for approximating the text-to-pattern hamming distance.
\newblock In Raimund Seidel, editor, {\em 1st Symposium on Simplicity in Algorithms ({SOSA} 2018)}, volume~61 of {\em OASIcs}, pages 10:1--10:5. Schloss Dagstuhl - Leibniz-Zentrum f{\"{u}}r Informatik, 2018.
\newblock \href {https://doi.org/10.4230/OASIcs.SOSA.2018.10} {\path{doi:10.4230/OASIcs.SOSA.2018.10}}.

\bibitem{LarsenW17}
Kasper~Green Larsen and R.~Ryan Williams.
\newblock Faster online matrix-vector multiplication.
\newblock In Philip~N. Klein, editor, {\em Proceedings of the Twenty-Eighth Annual {ACM-SIAM} Symposium on Discrete Algorithms, {SODA} 2017, Barcelona, Spain, Hotel Porta Fira, January 16-19}, pages 2182--2189. {SIAM}, 2017.
\newblock \href {https://doi.org/10.1137/1.9781611974782.142} {\path{doi:10.1137/1.9781611974782.142}}.

\bibitem{lee2002fast}
Lillian Lee.
\newblock Fast context-free grammar parsing requires fast boolean matrix multiplication.
\newblock {\em Journal of the ACM (JACM)}, 49(1):1--15, 2002.

\bibitem{li2019faster}
Jason Li.
\newblock Faster minimum k-cut of a simple graph.
\newblock In {\em 2019 IEEE 60th Annual Symposium on Foundations of Computer Science (FOCS)}, pages 1056--1077. IEEE, 2019.

\bibitem{lincoln2018tight}
Andrea Lincoln, Virginia~Vassilevska Williams, and Ryan Williams.
\newblock Tight hardness for shortest cycles and paths in sparse graphs.
\newblock In {\em Proceedings of the Twenty-Ninth Annual ACM-SIAM Symposium on Discrete Algorithms}, pages 1236--1252. SIAM, 2018.

\bibitem{Lovsz2007SzemerdisLF}
L{\'a}szl{\'o}~Mikl{\'o}s Lov{\'a}sz and Bal{\'a}zs Szegedy.
\newblock Szemer{\'e}di’s lemma for the analyst.
\newblock {\em GAFA Geometric And Functional Analysis}, 17:252--270, 2007.
\newblock URL: \url{https://api.semanticscholar.org/CorpusID:15201345}.

\bibitem{DBLP:journals/siamdm/Nagle10}
Brendan Nagle.
\newblock On computing the frequencies of induced subhypergraphs.
\newblock {\em {SIAM} J. Discret. Math.}, 24(1):322--329, 2010.
\newblock \href {https://doi.org/10.1137/090752961} {\path{doi:10.1137/090752961}}.

\bibitem{NP85}
Jaroslav Ne{\v{s}}et{\v{r}}il and Svatopluk Poljak.
\newblock On the complexity of the subgraph problem.
\newblock {\em Commentationes Mathematicae Universitatis Carolinae}, 26(2):415--419, 1985.

\bibitem{Patrascu10}
Mihai Pătraşcu.
\newblock Towards polynomial lower bounds for dynamic problems.
\newblock In Leonard~J. Schulman, editor, {\em 42nd Annual {ACM} Symposium on Theory of Computing ({STOC} 2010)}, pages 603--610. {ACM}, 2010.
\newblock \href {https://doi.org/10.1145/1806689.1806772} {\path{doi:10.1145/1806689.1806772}}.

\bibitem{RodittyZ04}
Liam Roditty and Uri Zwick.
\newblock On dynamic shortest paths problems.
\newblock In Susanne Albers and Tomasz Radzik, editors, {\em Algorithms - {ESA} 2004, 12th Annual European Symposium, Bergen, Norway, September 14-17, 2004, Proceedings}, volume 3221 of {\em Lecture Notes in Computer Science}, pages 580--591. Springer, 2004.
\newblock \href {https://doi.org/10.1007/978-3-540-30140-0\_52} {\path{doi:10.1007/978-3-540-30140-0\_52}}.

\bibitem{Strassen1969GaussianEI}
Volker Strassen.
\newblock Gaussian elimination is not optimal.
\newblock {\em Numerische Mathematik}, 13:354--356, 1969.

\bibitem{Vassilevska09Clique}
Virginia Vassilevska.
\newblock Efficient algorithms for clique problems.
\newblock {\em Inf. Process. Lett.}, 109(4):254--257, 2009.
\newblock \href {https://doi.org/10.1016/j.ipl.2008.10.014} {\path{doi:10.1016/j.ipl.2008.10.014}}.

\bibitem{VassilevskaWilliams18}
Virginia~Vassilevska Williams.
\newblock On some fine-grained questions in algorithms and complexity.
\newblock In {\em Proceedings of the International Congress of Mathematicians ({ICM} 2018)}, pages 3447--3487, 2018.
\newblock \href {https://doi.org/10.1142/9789813272880_0188} {\path{doi:10.1142/9789813272880_0188}}.

\bibitem{WilliamsW18}
Virginia~Vassilevska Williams and R.~Ryan Williams.
\newblock Subcubic equivalences between path, matrix, and triangle problems.
\newblock {\em J. {ACM}}, 65(5):27:1--27:38, 2018.
\newblock \href {https://doi.org/10.1145/3186893} {\path{doi:10.1145/3186893}}.

\bibitem{VassilevskaWilliamsX20}
Virginia~Vassilevska Williams and Yinzhan Xu.
\newblock Monochromatic triangles, triangle listing and {APSP}.
\newblock In Sandy Irani, editor, {\em 61st Annual {IEEE} Symposium on Foundations of Computer Science ({FOCS} 2020)}, pages 786--797. {IEEE}, 2020.
\newblock \href {https://doi.org/10.1109/FOCS46700.2020.00078} {\path{doi:10.1109/FOCS46700.2020.00078}}.

\bibitem{Yu18}
Huacheng Yu.
\newblock An improved combinatorial algorithm for boolean matrix multiplication.
\newblock {\em Inf. Comput.}, 261:240--247, 2018.
\newblock \href {https://doi.org/10.1016/j.ic.2018.02.006} {\path{doi:10.1016/j.ic.2018.02.006}}.

\end{thebibliography}


\end{document}